\documentclass[a4paper, 11pt]{article}
\usepackage[T1]{fontenc}
\usepackage[utf8]{inputenc}
\usepackage{latexsym}
\usepackage{amssymb,amsmath,amsthm}
\usepackage{graphicx}
\usepackage{wrapfig}

\newtheorem{tw}{Theory}[section]
\newtheorem{lm}{Lemma}[section]
\newtheorem{df}[tw]{Definition}
\usepackage{indentfirst}

\begin{document}

\begin{center}

{\LARGE \bf Parallel implementation of flow and matching algorithms}\\[20pt]

{\large Agnieszka Łupińska}\\
{\large Jagiellonian University, Kraków}\\
{\large agnieszka.lupinska@uj.edu.pl}\\[30pt]

{\bf Abstract}\\
\end{center}

{\it 
In our work we present two parallel algorithms and their lock-free implementations using a popular GPU environment Nvidia CUDA. The first algorithm is the push-relabel method for the flow problem in grid graphs. The second is the cost scaling algorithm for the assignment problem in complete bipartite graphs.
}

\section{Introduction}

The maximum flow problem has found many applications in various computer graphics and vision problems. For example, the algorithm for~the~graph cut problem is an optimization tool for the optimal MAP (Maximum A-Posteriori Probability) estimation of energy functions defined over an MRF (Markov Random Field) [4, 11, 12, 13, 17]. Another new and most interesting for us idea consists in computing optical flow by reducing it to~the~assignment (weighted matching) problem in bipartite graphs ([18]). Therefore it is important to look for new solutions to improve the execution time of~algorithms solving the max flow and related problems. 

The~new~approach to acceleration of algorithms uses the power of GPU (graphics parallel units). This has motivated us to pose the main purpose of~this~research: find an efficient parallel algorithm for the weighted matching problem and implement it in a popular GPU environment Nvidia CUDA. Naturally, such an algorithm can use max flow computation techniques and~the~easiest method to compute maxflow in parallel is the push-relabel algorithm.
Hence in the first part of our work we develop our own CUDA implementation of Hong's lock-free push-relabel algorithm. It can be used to find graph cuts in graphs constructed by Kolmogorov et.~al.~in~[12] and is suitable for minimization of any energy function of the class characterized by these autors.

\begin{figure}[h!b!]
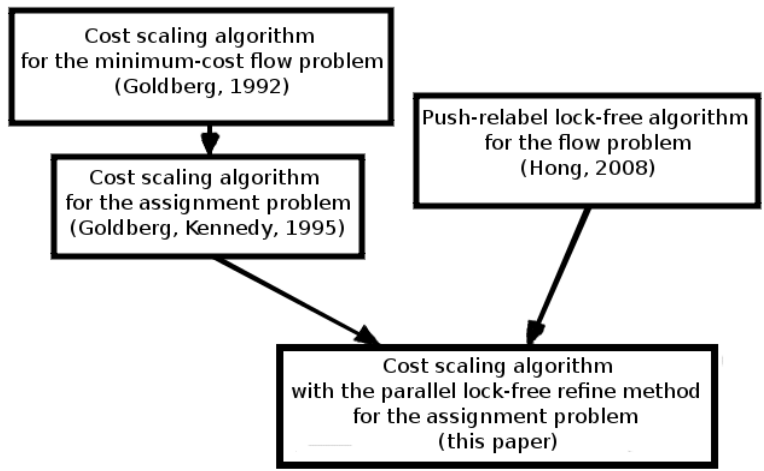

\begin{center}
\includegraphics[width=4in]d
\end{center}
\caption{The algorithms described in Section 5}
\end{figure}

In the second part of the paper we focus on the assignment problem, ie.~the~max weight matching problem. The cost scaling algorithm solving this problem has also found applications in many computer vision tasks such as~recognition and tracking of images.
Finally we present our implementation of the cost scaling algorithm for the assignment problem, where the core \textit{refine} procedure is implemented lock-free on CUDA.

This work is organized as follows. Section 2 contains all needed definitions in our paper. Section 3 introduces the CUDA programming environment. Section 4 presents the sequential push-relabel algorithm and its parallel counterparts: a blocking version of the parallel push-relabel algorithm, presented by Vineet and~Narayanan [4], and a lock-free push-relabel algorithm by Hong [5]. In the end of this section we present our CUDA implementation of the lock-free push-relabel algorithm for grid graphs. 
Section 5 starts from the presentation of two sequential algorithms: a scaling minimum-cost flow method and its counterpart for the assignment problem, both developed by Goldberg et al. [2,8,9]. Then we present our own parallel cost scaling algorithm for the assignment problem using the lock-free push-relabel algorithm. In the end of section 5 we present our CUDA implementation of this algorithm for arbitrary graphs. Figure 1 shows the diagram of the reductions between the main analyzed problems.

\section{CUDA Programming}

CUDA (Compute Unified Device Architecture) is a parallel computing architecture for Nvidia GPUs. In the past the GPU was used only as the~coprocessor of the CPU to reply on the many graphics tasks in real-time. Now, thanks to increasing computation power of GPUs, they are also very efficient on many data-parallel tasks. Therefore GPUs are used also for many non-graphics applications like the push-relabel max flow algorithm. 

We present two algorithms implemented in CUDA 4.0. Both were tested on Nvidia GTX 560 Ti (i.e. on a device of compute capability 2.1). 

The programs operating on the GPU are implemented in the CUDA C~language, an extention of C. CUDA C allows to define C functions called \textit{kernels} that can be executed in parallel. The C program is run on the \textit{host} (or CPU) by the host thread. The kernels (CUDA programs) are launched by the host thread and run on the \textit{device} (or GPU) by many CUDA threads. 
The~number of threads executing a kernel is defined by the programmer. The~running threads are split into three-dimensional blocks. The set of~blocks forms a three-dimesional grid.\\
An example of a kernel declaration is:

\begin{quote}
\begin{scriptsize}
\_\_global\_\_ void kernel\_function(int* data);
\end{scriptsize}
\end{quote}
which must by called like this: 
\begin{quote}
\begin{scriptsize}
dim3 gD $= 10$; /* dimension of grid, unspecified component is initialized to 1*/ \\
dim3 bD $= (32, 8, 1)$; /* dimension of block */ \\
kernel\_function $\lll$gD, bD$\ggg$ (someArray);\\
\end{scriptsize}
\end{quote}

Each thread executing a kernel is given the following built-in three-dimensional variables:
\begin{quote}
\begin{scriptsize} 
dim3 threadIdx $ = ($ threadIdx.x, threadIdx.y, threadIdx.z $)$ /* unique for each thread in the block */\\
dim3 blockIdx $ = ($ blockIdx.x, blockIdx.y, blockIdx.z $)$	/* unique for each block in grid */\\
dim3 blockDim $ = ($ blockDim.x, blockDim.y, blockDim.z $)$	\\
dim3 gridDim $ = ($ gridDim.x, gridDim.y, gridDim.z $)$
\end{scriptsize}
\end{quote} 
Hence for each thread we can calculate its unique index in grid in the following way:
\begin{quote}
\begin{scriptsize}
int threadsInBlock $=$ blockDim.x $*$ blockDim.y;\\
int numberOfBlockInGrid $= ($blockIdx.y $*$ gridDim.x$)\ + $ blockIdx.x;\\
int numberOfThreadInBlock $= ($threadIdx.y $*$ blockDim.x$)\ + $ threadIdx.x;\\
int thid\_id $= ($threadsInBlock $*$ numberOfBlockInGrid$)\ +$ numberOfThreadInBlock;
\end{scriptsize}
\end{quote}

The host and the device have separated memory spaces called the \textit{host memory} and the \textit{device memory}, both residing in the dynamic random-access memory (DRAM). There are several types of device memory: \textit{global}, \textit{local}, \textit{shared}, \textit{constant} and \textit{texture}. There are also \textit{registers}. The shared memory and the registers are located on the chip of the GPU. The others are located off the chip so they have large access latency. However the sizes of the shared memory and the registers are much smaller than of the others. Note that~since the local memory is located off chip then its access latency is also big.

The scopes and the lifetimes of the local memory and the registers are~restricted to one thread. The scope of the shared memory and its lifetime is restricted to all threads of the block. The lifetime and access to other of~memory are avaliable for all launched threads and the host.

For devices of compute capability 2.x the local and the global memory are~cached. There are two types of cache: an L1 cache for each multiprocessor and an L2 cache shared by all multiprocessors on the GPU. The size of~the~L2 cache is fixed equal to 768 KB. The L1 cache and~the~shared memory are~stored in the same on-chip memory and their initial sizes are~48 KB of shared memory and 16 KB of L1. These values can be reversed by the~function \textit{cudaFuncSetCacheConfig()} (or \textit{cuFuncSetCacheConfig()} for~Driver~API). In~our implementations the shared memory is unuseful and its size that we use is not bigger than 16 KB. However, we use the first configuration because it~gives better running times. 

The threads can be synchronized in the scope of a block by the function \mbox{\textit{\_\_syncthreads()}}. It sets a semaphore which causes that the execution of~the~further code waits until all parallel threads reach the specified point. 
	
In our implementations we use atomic functions: \textit{atomicAdd}() and \textit{atomicSub}(), avaliable for devices of compute capability 2.x, which perform a~read-modify-write operations on 64-bit words residing in the global memory. The~atomic operations are slower than their non-atomic counterparts but they allow us to implement the programs without any synchronization \\of~the~threads. 

A bandwidth is the rate at which data can be transferred. The bandwidth between the global memory on device and~the global memory on the~host is much smaller than the bandwidth between the global memory on device and~the~memory space on GPU. Therefore it is important to minimize the~data transfer between the device and the host. In our implementation we strived to reduce the copying only to necessary arrays of data.

To allocate and deallocate memory on the device we use the functions \textit{cudaMalloc()} and \textit{cudaFree()} and to copy memory between the device\\ and~the~host we use the function \textit{cudaMemcpy()}. 

In our implementations we use the cutil.h library available in the GPU computing SDK which allows us to detect the errors returned by the device.

\section{Basic Graph Definitions}

Let $G=(V,E)$ be a directed graph and $u: E \rightarrow \mathbf{R^{+}}$ be an edge capacity function. 
For formal reasons, if $(x, y) \notin E$ then we set $u(x, y) = 0$.
Let $s \neq t$ be two distinguished verteces in $V$, the source and the sink, respectively. Then the triple $F=(V, E, u, s, t)$ is called a \textbf{flow network}.

A \textbf{pseudoflow} is a function $f: V\times V \rightarrow \mathbf{R}$ such that, for each $(x, y)~\in~E$, $f$ satisfies the capacity constraints: $f(x,y) \leq u(x, y)$, and~the~skew symmetry: $f(x,y) = -f(y,x)$.
We say that the pseudoflow $f$ is a \textbf{flow} if for each node $x \in V - \{s, t\}$, $\sum_{(x, y) \in E}{f(x, y)} = 0$. 
We say that~an~edge~$(x, y)$ is in the flow $f$, if $f(x, y) > 0$. The \textbf{value of the flow} f is $|f| = \sum_{(s, x) \in E}{f(s, x)}$.

In the \textbf{max flow problem} we are given a flow network $F$ with a capacity function $u$ and distinguished vertices $s, t$. We must find a flow $f$ from $s$ to~$t$, such that $|f|$ is maximum.

The \textbf{residual capacity} of an edge \mbox{$(x,y) \in E$} is $u_f(x, y) = u(x, y) - f(x, y)$. 
An edge, for which $u_f(x, y) > 0$, is called a \textbf{residual edge}. The~set of all residual edges in $G$ is~denoted by $E_f$. The graph $G_f = (V_f, E_f)$ is~called a \textbf{residual graph}.

It is easy to see that if $u_f(x, y) > 0$ then $u_f(y, x) > 0$ and it is possible to push more units of flow through $(y, x)$. 
Let $u_f(x, y) > 0$ for $(x, y) \in E$. If~$f(x, y) > 0$ than $f(y, x) < 0$. This implies that $u_f(y, x) = u(y, x) - f(y, x)~>~0$ because of $u(y, x) \geq 0$. 

To present the push-relabel algorithm we need to introduce the following definitions. Let $F$ be a flow network and $f:V\times V \rightarrow \mathbf{R}$ be a pseudoflow function. For each $x \in V$ we define $e(x)$, the \textbf{excess} of the node $x$, as~the~sum: $e(x) = \sum_{(z, x) \in E}{f(z, x)} - \sum_{(x, y) \in E}{f(x, y)}$. Note if $f$ is a flow then $e(x) = 0$ for all $x \in V$.~If~$e(x) > 0$, for some node $x \in V$, then we say that $x$ is an \textbf{active node}. The \textbf{height} of a node $x$, denoted $h(x)$ is a~natural number from $[0, n]$.

Now we give definitions related to the matching problems and max flow min cost algorithms. Let $G\ =(V=X\cup Y,\ E)$ be a bipartite graph, $|X| = |Y| = n, |E| = m$. The bipartite \textbf{matching problem} is to find a~largest cardinality subset $M \subseteq E$ such that each vertex $x$ belongs to~at~most one edge of $M$. In case $|M| = n$, $M$ is called a \textbf{perfect matching}.

If $w: E \rightarrow \mathbf{R}$ is a weight function for the edges then $w(M) = \sum_{(x,y) \in M}{w(x, y)}$ is the weight of the matching $M$. 
The \textbf{assignment problem} is to find the~largest cardinality matching $M \subseteq E$ of the maximum weight $w(M)$.

For any graph $G=(V, E)$ the \textbf{cost} function is $c: E \rightarrow \mathbf{R}$. Hence, the~$\textbf{cost of a pseudoflow}$ $f$ is $c(f) = \sum_{(x, y) \in E}{c(x, y)f(x,y)}$.

Assume a flow network $F = (V, E, u, s, t)$ is given with an aditional edge cost function $c$. The \textbf{max flow min  cost problem} consist in finding a~maximum flow with the lowest possible cost.

Following [9] we introduce definitions used in the cost scaling algorithm, they extend the notation given above. 

Let $I' = (V, E, u, s, t, c)$ be a instance of the max flow min cost problem. A \textbf{price} of a node is given by a function $p: V \rightarrow \mathbf{R}$.
A \textbf{reduced cost} of~an~edge $(x, y) \in E$ is $c_p(x, y) = c(x, y) + p(x) - p(y)$, while a \textbf{part-reduced cost} of an edge $(x, y) \in E$ is $c_p'(x, y) = c(x, y) - p(y)$.

For a constant $\epsilon \geq 0$ a pseudoflow $f$ is \textbf{$\epsilon$-optimal} with respect to a price function $p$ if, for every residual edge $(x, y) \in E_f$, we have $c_p(x, y) \geq -\epsilon$.
For~a~constant $\epsilon \geq 0$ a pseudoflow $f$ is \mbox{$\epsilon$-optimal} if it is $\epsilon$-optimal with~respect to some price function $p$.

We say that a residual edge \mbox{$(x, y) \in E_f$} is \textbf{admissible} if $c_p(x, y) < 0$.

\section{Push-Relabel Algorithm}

	\subsection{Sequential Algorithm}

In this section we focus on a standard sequential algorithm solving the~max flow problem.

Among many sequential algoritms solving the max flow problem, there~are three basic methods. 
The \textit{Ford-Fulkerson} and the \textit{Edmonds-Karp} algorithms are the most common and easiest. The Ford-Fulkerson algorithm calculates the flow in time $O(E|f^*|)$, where $|f^*|$ is value of the max flow $f^*$. The running time of the Edmonds-Karp algorithm is $O(VE^2)$. In every step, these~two algorithms look for a new augumenting path from the source to~the~sink. If found, they increase the flow on the edges of the path. Otherwise they stop and return the flow found up to this moment (it is optimal). 

The third solution is \textit{the push-relabel} algorithm. 
We present its generic version whose time complexity is $O(V^2E)$. Next we discribe two heuristics which significantly improve its execution time. The generic push-relabel version with two heuristics can be effectivly parallalized and provides a basis for our further considerations.

At the beginning of algorithm $h(x) = 0$, for all $x \in V-\{s\}$, and~$h(s)~=~|V|$. We have also $e(x) = 0$, for all $x \in V$, and $e(s) = \infty$ (Algorithm 4.1).\\

\begin{scriptsize}
\begin{tabular}{l} 

\hline

\textbf{Algorithm 4.1. \textit{Init} operation for the push-relabel algorithm} \\

\hline\\

for each $(x, y) \in E$ do \\% - \{s\} \times V$ do \\
\hspace{0.5cm}
$f(x, y) \leftarrow 1 $\\
\hspace{0.5cm}
$f(y, x) \leftarrow 0 $\\
%for each $(x, y) \in \{s\} \times V$ do \\
%\hspace{0.5cm}
%$f(s, y) \leftarrow 0 $\\
%\hspace{0.5cm}
%$f(y, s) \leftarrow 1 $\\
for each $x \in V - \{s\}$ do\\
\hspace{0.5cm}
$e(x) \leftarrow 0 $\\
\hspace{0.5cm}
$h(x) \leftarrow 1 $\\
$e(s) \leftarrow \infty $\\
$h(s) \leftarrow |V|$\\[3pt]

\hline

\end{tabular}
\end{scriptsize}
\\
\\
As opposed to the previous two algorithms, \textit{push-relabel} does not look for~augumenting paths but acts locally on the nodes (Algorithm 4.2).
\\
\\
\begin{scriptsize}
\begin{tabular}{l} 

\hline

\textbf{Algorithm 4.2. The \textit{push-relabel} algorithm} \\

\hline\\

$Init()$\\
make set $S$ empty\\
$S \leftarrow s$\\
while ($S$ is not empty) do\\
\hspace{0.5cm} $x \leftarrow S$.pop()\\
\hspace{0.5cm} $discharge(x)$\\
\hspace{0.5cm} if ($x$ is active node)\\
	\hspace{1.2cm} $S$.push(x)\\[3pt]

\hline

\end{tabular}
\end{scriptsize}
\\
\\

The \textit{discharge} operation for each active node $x$ with the set $S$ selects either \textit{push} or \textit{relabel} operation (Algorithm 4.3).\\

\begin{scriptsize}
\begin{tabular}{l} 

\hline

\textbf{Algorithm 4.3. \textit{discharge(x)}, \textit{push(x, y)} and \textit{relabel(x)} operations} \\

\hline\\

\textit{discharge(x)}:\\[3pt]
if ($\exists{\ (x, y) \in E_f}:\ h(x) = h(y) + 1$)\\[2pt]
	\hspace{0.5cm} $push(x, y)$\\
else \\
	\hspace{0.5cm} $relabel(x)$
\\[5pt]

\textit{push(x, y)}:\\[3pt]
$\delta \leftarrow \min{\{u_f(x, y),\ e(x)\}}$\\
$e(x) \leftarrow e(x) - \delta$\\
$e(y) \leftarrow e(y) + \delta$\\
$f(x, y) \leftarrow f(x, y) + \delta$\\
$f(y, x) \leftarrow f(y, x) - \delta$\\
\\

\textit{relabel(x)}:\\[3pt]
$h(x) \leftarrow \min{\{h[y]:\ (x, y) \in E_f\}} + 1$\\
\\[3pt]\hline

\end{tabular}
\end{scriptsize}
\\
\\

The push operation is performed on an active node $x$, for which there exists an outgoing residual edge $(x, y) \in E_f$ and the node $y$ satisfies the~height constraint: $h(x)~=~h(y)~+~1$. If the node $x$ is active and every edge $(x, y) \in E_f$ does not satisfy this constraint then the relabel operation is~performed.

It can be shown that the generic push-relabel algorithm is correct, terminates and its running time is $O(V^2E)$. The proofs and a comprehensive discussion about the generic push-relabel algorithm and its improvments, can be found in [1] and [2].

		\subsection{Heuristics: Global and Gap Relabeling}

The above version of the push-relabel algorithm has poor performance in~practical applications. To improve the running time two heuristics are used: global relabeling and gap relabeling [2]. To get intuition how these heuristics work, we make some observations.

During the execution of the algorithm, both the excess and the height of~the~source and the sink are not changing. Then, to the end of the algorithm, the source height is $|V|$ and the sink height is $0$. Let us define a~\textit{distance function}.

\begin{df}[\textbf{distance function}]
Let $G=(V, E)$ be a network with~a~flow $f$ and let $E_f$ be the set of its residual edges. The function $h: V \rightarrow \mathbf{N}$ is a distance function if $h(s) = |V|$, $h(t) = 0$, and for each edge $(x, y)~\in~E_f,\\ h(x)~\leq~h(y) + 1$.
\end{df}

It can be proved ([1]) that the height function satisfies the properties of~a~distance function at each step of the algorithm. Hence we can think about the node height as its distance from the sink to the source. The major issue from which the algorithm's performance suffers is an execution of~a~lot of~unnecessary relabel operations. It can be proved ([1]) that during the~execution, a node height can reach a limit of $2|V|-1$. The \textit{global relabeling} heuristic (Algorithm 4.4) prevents the heights of the nodes from growing fast and assigns them the smallest admissible heights. 
\\
\\
\begin{scriptsize}
\begin{tabular}{l}

\hline

\textbf{Algorithm 4.4 \textit{global relabeling} heuristic} \\

\hline\\

make $Q$ empty queue\\
for each $x \in V$\\
\hspace{0.5cm}$x$ scanned $\leftarrow$ false\\
$Q$.enqueue($t$)\\
$t$.scanned $\leftarrow$ true\\
while ($Q$ not empty) do\\
\hspace{0.5cm}$x \leftarrow Q$.dequeue()\\
\hspace{0.5cm}$current \leftarrow h(x)$\\
\hspace{0.5cm}$current \leftarrow current + 1$\\
\hspace{0.5cm}$\forall\ (y, x) \in E_f$:\ $y$ not scanned do\\
	\hspace{1.2cm}$h(y) \leftarrow current$\\
	\hspace{1.2cm}$y$ scanned $\leftarrow$ true\\
	\hspace{1.2cm}$Q$.enqueue(y)\\
\\[3pt]\hline

\end{tabular}
\end{scriptsize}
\\
\\

The global relabeling technique consists in performing a breadth-first backwards search (BFS) in the residual graph and assigning the new heights, equal to the level number of a node in the BFS tree. Obviously BFS takes linear time $O(m+n)$. Usually global relabeling is performed once every $n$~relabels. This heuristic significantly improves the performance of the push-relabel method.

The second heuristic is \textit{gap relabeling}. The gap relabeling "removes" from the residual graph the nodes that will never satisfy the height constraint. This improves the performance of the push-relabel method because of reducing the number of active nodes. However its result is not so significant for the running time as the global relabeling. The gap relabeling heuristic also can be done in linear time.

The gap relabeling can be performed after the BFS loop of the global relabeling. Any non-scanned node $x$ is not reachable from the sink so we can set its height to $|V|$. This makes the pushed flow omit $x$ and go to another node $y$ (from which there is an augumenting path). As a result the pushed flow gets faster to the sink.

We have presented the generic push-relabel algorithm with two additional techniques which will be used next in parallel versions. Further improvements of the sequential push-relabel algorithm can be found e.g. in [1], [2] and [3].

	\subsection{Parallel Approach}

The push-relabel algorithm was parallelized by Anderson and Setubal in~1992~[14]. One of the first CUDA implementation was proposed by Vineet and Narayanan ([4]). 
They presented the push-relabel algorithm to~graph cuts on the GPU, which is a tool to find the optimal MAP estimation of~energy functions defined over an MRF [11, 12, 13]. 

Vineet and Narayanan's CUDA implementations were tested on the Nvidia 280 and 8800 GTX graphic cards (devices of compute capability at most 1.3). Their algorithm works on grid graphs which arise in MRFs defined over images. The dataset and CUDA implementations are available\\ from~http://cvit.iiit.ac.in/index.php?page=resources. 

The authors assumed that each node of a graph is handled by one thread and the number of outgoing edges per node is fixed, equal 4. Authors suggested that the algorithm can be implemented for the expanded 3D graphs, that is with 8 outgoing edges per node. Their algorithm requires a computer architecture that could launch the same number of threads as~the~number of~the~graph nodes so possibilities to run this algorithm on a CPU are small.

The authors prepared two implementations of push-relabel algorithm: atomic and non-atomic. The first of them requires two phases: \textit{push} and~\textit{relabel}. The second implementation, which was designed for devices of compute capability lower than 1.2, additionally requires a \textit{pull} phase. Further we will describe only the first implementation, more details about the second can be found in [4]. 

Vineet and Narayanan have used the graph construction of Kolmogorov et. al. ([12]) which maintains the grid structure, suitable for the CUDA architecture. The data of a grid graph are stored in the global and the shared memory of the device. They are in 8 separated tables. Table of~the~heights is~stored in the shared memory and other tables are stored in the global memory. Among them are the excesses, the relabel masks (which say whether the~node is active), the residual capacities of edges upwards/downwards nodes, the residual capacities of edges towards the nodes on the left/right and the residual capacities of edges toward the sink. Access to the element in table is by calculating its index. 

Before running the algorithm, the host thread copies data to the global memory on device. The main loop of algorithm is executed on~CPU and~for~each phase the host thread calls another kernel. After finishing a push kernel, the~control is returned to the host thread and can be launch the next relabel kernel. However, authors suggested that for some grid graphs, running $m$ push phases before each relabel phase, improved the execution time. Algorithm stops when all excesses stay the same after a few iterations of loop.

In first step of the push kernel, each node saves its height in the shared memory of thread-block. 
After the saving, each node whose relabel mask is set to 1 pushes the flow toward its neighbors, if they satisfy a height constraint. In this step the threads read the heights saved in shared memory. 

In the relabel phase first, each thread sets a new relabel mask. If~excess of~the~node is positive and is connected with neighbor, which satisfied a~height constraint, the mask is~set to~one (it is active). If node only has positive excess, the mask is set to zero (node is passive). Otherwise, mask is~set to two (inactive node) and this node will never be active. After setting all the relabel masks, threads can calculate new heights of the nodes. They read the old heights saved in the shared memory and write the~new to~the~global memory. It is preparation for next a push phase.

To improve an execution time of algorithm, Authors experimented with~varying numbers of threads per block. The best result, they obtained for~a~thread-block of size $32 \times 8$.

In the Vineet et al. implementation, the synchronization of~threads is~assured by \_\_syncthreads() CUDA function. This approach to the parallel push-relabel algorithm blocks the threads which are ready to run next operation before finishing other threads. Therefore it causes long execution time of the algorithm.

	\subsection{Parallel Lock-Free Algorithm}

In 2008 Hong presented a lock-free multi-threaded algorithm for the max flow problem ([5]) based on Goldberg's version ([2]) of the push-relabel algorithm. Implementation of Hong's algorithm requires a multi-threaded architecture that supports read-modify-write atomic operations. In our implementation we have used a Nvidia CUDA atomicAdd(int*, int, int) and~atomicSub(int*, int, int) functions. 

Without lost of generality we assume that the number of running threads is $|V|$ and each of them handles exactly one node of the graph, including all push and relabel operations on it. In several, a few nodes can be handled by~one~thread. 

Let $x$ be the running thread representing the node $x \in V$. In Hong's algorithm (Algorithm 4.5) each of the running threads has the following private atributes. The variable $e'$ stores the excess of the node $x$. The~variable $h'$ stores the height of the currently considered neighbour $y$ of $x$ such that $(x, y) \in E_f$. The~variable $\tilde{h}$~stores the height of the lowest neighbour $\tilde{y}$~of~$x$.

Other variables are shared between all the running threads. Among them there are the arrays with excesses and heights of nodes, and residual capacities of edges.

First, the \textit{Init} operation is performed by the master thread, in CUDA programing it is the host thread. This init code is the same as its counterpart in the sequential push-relabel version. Next, the master thread starts the~threads executing in parallel the lock-free push-relabel algorithm (in CUDA the host thread launches kernels).
	  
The basic changes, introduced by Hong, deal with the selection of operation (push or relabel) that should be executed by $x$, and to which of~the~adjacent nodes $\tilde{y}, (x, \tilde{y}) \in E_f$, a flow must be pushed. In~opposite to~the~push operation of the generic sequence version where any node $y$ connected by~a~residual edge to $x$ such that $h(x) = h(y) + 1$ could be pushed, it selects the lowest node among all the nodes connected by a residual edges (lines 4-9). Next if~the~height of $\tilde{y}$ is less than the height of $x$ (line 10), the push operation is performed (lines 11-15). Otherwise, the relabel operation is~performed, that is, the height of $x$ is modified to $h(\tilde{y}) + 1$ (line 17). Note that the~relabel operation need not be atomic because only the $x$ thread can change the value of the height of $x$. Furthermore, all critical lines in the code where more than two threads execute the write instruction are atomic. Hence it is easy to see that the algorithm is correct in respect to read and write instructions.
\\
\\
\begin{scriptsize}
\begin{tabular}{l}

\hline

\textbf{Algorithm 4.5. Lock-free multi-threaded push-relabel algorithm by Bo Hong} \\

\hline\\

\textit{Init():}\\[3pt]
$h(s) \leftarrow |V|$\\
for each $x \in V - {s}$\\
\hspace{0.5cm}$h(x) \leftarrow 0$\\
for each $(x, y) \in E$\\
\hspace{0.5cm}$u_f(x, y) \leftarrow u_{xy}$\\
\hspace{0.5cm}$u_f(y, x) \leftarrow u_{yx}$\\
for each $(s, x) \in E$\\
\hspace{0.5cm}$u_f(s, x) \leftarrow 0$\\
\hspace{0.5cm}$u_f(x, s) \leftarrow u_{xs} + u_{sx}$\\
\hspace{0.5cm}$e(x) \leftarrow u_{sx}$\\[10pt]

\textit{lock-free push-relabel():}\\[3pt]
1. /*x - node operated by the x thread*/\\
2. while ($e(x) > 0$) do\\
3. \hspace{0.5cm}$e' \leftarrow e(x)$\\
4. \hspace{0.5cm}$\tilde{y} \leftarrow NULL$\\
5. \hspace{0.5cm}for each $(x, y) \in E_f$ do\\
6.        \hspace{1.2cm}$h' \leftarrow h(y)$\\
7.        \hspace{1.2cm}if $\tilde{h} > h'$ do\\
8.        	\hspace{2.0cm}$\tilde{h} \leftarrow h'$\\
9.		\hspace{2.0cm}$\tilde{y} \leftarrow y$\\
10.\hspace{0.5cm}if $h(x) > \tilde{h}$ do /*the x thread performs PUSH towards y*/\\
11.	\hspace{1.2cm}$\delta \leftarrow min\{e',\ c_f(x, \tilde{y})\}$\\
12.	\hspace{1.2cm}$c_f(x, \tilde{y})\ \leftarrow c_f(x, \tilde{y}) - \delta$\\
13.	\hspace{1.2cm}$c_f(\tilde(y), x)\ \leftarrow c_f(\tilde(y), x) + \delta$\\
14.	\hspace{1.2cm}$e(x) \leftarrow e(x) - \delta$\\
15.	\hspace{1.2cm}$e(\tilde(y)) \leftarrow e(\tilde(y)) + \delta$\\
16.\hspace{0.5cm}else do /*the x thread performs RELABEL*/\\
17.	\hspace{1.2cm}$h(x) \leftarrow \tilde{h} + 1$\\
\\[3pt]\hline

\end{tabular}
\end{scriptsize}
\\
\\

All running threads have access to a critical variable in the global memory, but indeed their task are peformed sequentialy thanks to the atomic access to the data. The order of the operations in this sequence cannot be predicted. Despite this the algorithm can be proved correct.

It can be proved that the lock-free algorithm terminates after at most $O(V^2 E)$ the push/relabel operations. Because it is executed in parallel by~many threads, the complexity of algorithm is analyzed in the number of the operiations, not in the execution time.

	\subsection{Lock-Free Algorithm and Heuristics}

In 2010 Hong and He ([6]) improved the lock-free push-relabel algorithm, by adding to it a sequential global relabel heuristic performed on CPU (Algorithm 4.6). 
According to the authors and their experimental result, the new CPU-GPU-Hybird scheme of the lock-free push-relabel algorithm is robust and efficient. 
\\
\\
\begin{scriptsize}
\begin{tabular}{l}

\hline

\textbf{Algorithm 4.6. CPU-GPU-Hybird of push-relabel algorithm by Hong and He} \\

\hline\\

\textit{Init():}\\[3pt]
1. initialize $e$, $h$, $u_f$ and \textit{ExcessTotal}\\
2. copy $e$ and $u_f$ from the CPU main memory to the CUDA global memory\\[10pt]

\textit{push-relabel-cpu():}\\[3pt]
1. while ($e(s) + e(t) < ExcessTotal$) do\\
2. \hspace{0.5cm}copy $h$ from the CPU main memory to the CUDA global memory\\
3. \hspace{0.5cm}call \textit{push-relabel-kernel()}\\
4. \hspace{0.5cm}copy $u_f$, $h$ and $e$ from CUDA global memory to CPU main memory\\
5. \hspace{0.5cm}call \textit{global-relabel-cpu()}
\\[3pt]\hline

\end{tabular}
\end{scriptsize}
\\
\\

Following the Authors we will talk about the CPU-GPU-Hybrid algorithm in the context of CUDA programming. 
Similarly as in the previous algorithm we assume that each node is operated by at most one thread. 
The~previous version of the~lock-free push-relabel algorithm will be called the generic algorithm, while the CPU-GPU-Hybrid scheme will be named the~hybrid algorithm.

The initialization of the hybrid algorithm is the same as its counterpart in the generic algorithm. 
The hybrid algorithm maintains 3 arrays with excesses, heights and residual capacities of the nodes and the edges and keeps the global variable $ExcessTotal$, equal to the value of the flow pushed from the source. 
$ExcessTotal$ resides in the global memory on the host and can be changed during the global relabeling. 

In opposite to the generic algorithm, the main body of the hybrid algorithm is controlled by the host thread on CPU. The host thread executes the while loop until the cumulative value of the excesses stored in the source and the sink achieves the value of \textit{ExcessTotal}. In this moment all the valid flow gets to the sink, and the rest of a flow returns to the source. Then, the~excess at~the~sink equals the value of the maximum flow.

In the first step of the while loop, the host thread copies the~heights of~the~nodes to device and launches the \textit{push-relabel-kernel}. 
When the control is returned back to the host thread, the calculated pseudoflow and the~heights of the nodes are copied to the CPU memory and the \textit{global-relabel-cup} is performed.
\\
\\
\begin{scriptsize}
\begin{tabular}{l}

\hline

\textbf{Algorithm 4.7. Initialization for CPU-GPU-Hybrid} \\

\hline\\[3pt]
\textit{Init():}\\[3pt]
1. $h(s) \leftarrow |V|$\\
2. $e(s) \leftarrow 0$\\
3. for each $x \in V - {s}$\\
4. \hspace{0.5cm}$h(x) \leftarrow 0$\\
5. \hspace{0.5cm}$e(x) \leftarrow 0$\\
6. for each $(x, y) \in E$\\
7. \hspace{0.5cm}$u_f(x, y) \leftarrow u_{xy}$\\
8. \hspace{0.5cm}$u_f(y, x) \leftarrow u_{yx}$\\
9. for each $(s, x) \in E$\\
10.\hspace{0.5cm}$u_f(s, x) \leftarrow 0$\\
11.\hspace{0.5cm}$u_f(x, s) \leftarrow u_{xs} + u_{sx}$\\
12.\hspace{0.5cm}$e(x) \leftarrow u_{sx}$\\
13.\hspace{0.5cm}$ExcessTotal \leftarrow ExcessTotal + u_{sx}$\\[3pt]

\hline\\

\end{tabular}
\end{scriptsize}
\\
\\
\begin{scriptsize}
\begin{tabular}{l}

\hline

\textbf{Algorithm 4.8. Lock-free push-relabel and global relabel for CPU-GPU-Hybrid}\\

\hline\\[3pt]
\textit{lock-free push-relabel():}\\[3pt]
1. /*x - node operated by the x thread*/\\
2. while ($CYCLE > 0$) do\\
3. \hspace{0.5cm}if ($e(x) > 0$ and $h(x) < |V|$) do\\
4. 	\hspace{1.2cm}$e' \leftarrow e(x)$\\
5. 	\hspace{1.2cm}$\tilde{y} \leftarrow NULL$\\
6. 	\hspace{1.2cm}for each $(x, y) \in E_f$ do\\
7.        	\hspace{2.0cm}$h' \leftarrow h(y)$\\
8.		\hspace{2.0cm}if $\tilde{h} > h'$ do\\
9.              	\hspace{2.7cm}$\tilde{h} \leftarrow h'$\\
10.              	\hspace{2.7cm}$\tilde{y} \leftarrow y$\\
11.	\hspace{1.2cm}if $h(x) > \tilde{h}$ do /*then the x thread perform PUSH towards y*/\\
12.	        \hspace{2.0cm}$\delta \leftarrow min\{e',\ c_f(x, \tilde{y})\}$\\
13.		\hspace{2.0cm}$c_f(x, \tilde{y})\ \leftarrow c_f(x, \tilde{y}) - \delta$\\
14.		\hspace{2.0cm}$c_f(\tilde(y), x)\ \leftarrow c_f(\tilde(y), x) + \delta$\\
15.		\hspace{2.0cm}$e(x) \leftarrow e(x) - \delta$\\
16.		\hspace{2.0cm}$e(\tilde(y)) \leftarrow e(\tilde(y)) + \delta$\\
17.	\hspace{1.2cm}else do /*then the x thread perform RELABEL*/\\
18.     	\hspace{2.0cm}$h(x) \leftarrow \tilde{h} + 1$\\
19.\hspace{0.5cm}$CYCLE \leftarrow CYCLE - 1$\\[10pt]

\textit{global relabeling heuristic():}\\[3pt]
1. for all $(x, y) \in E$ do\\
2. 	\hspace{0.5cm}if ($h(x) > h(y) + 1$) then\\
3. 		\hspace{1.2cm}$e(x) \leftarrow e(x) - u_f(x, y)$\\
4. 		\hspace{1.2cm}$e(y) \leftarrow e(y) + u_f(x, y)$\\
5. 		\hspace{1.2cm}$u_f(y, x) \leftarrow u_f(y, x) + u_f(x, y)$\\
6. 		\hspace{1.2cm}$u_f(x, y) \leftarrow 0$\\
7. do a backwards BFS from the sink and assign the height function\\
8. with each node's BFS tree level\\
9. if (not all the nodes are relabeled) then\\
10.	\hspace{0.5cm}$\forall x \in V$ do\\
11.		\hspace{1.2cm}if ($x$ is not relabeled and marked) then\\
12.        		\hspace{2.0cm}mark $x$\\
13.	   		\hspace{2.0cm}$ExcessTotal \leftarrow ExcessTotal - e(x)$\\

\\[3pt]\hline

\end{tabular}
\end{scriptsize}
\\
\\
\\

The \textit{push-relabel-kernel} algorithm differs from the generic algorithm\\ in~the~timing of~executing the kernel. The thread stops the while loop after $CYCLE$ iterations (where $CYCLE$ is an integer constant defined by~the~user) and not when its node becomes inactive. After stopping the~loop the~heuristic is called, and then the loop is initialized again.

Since the while loop can terminate at any moment (randomly in respect to the original sequential flow computation) it may occur that the property of some residual edge $(x, y) \in E_f$ is violated, i.e. $h(x) > h(y) + 1$.
Then, before computing the new heights of nodes, all the violating edges must be canceled by pushing the flow. It is made in lines 1-6 of the \textit{global relabeling heurisitic}.
Next, the nodes are assigned new heights by performing in the~residual graph a backwards BFS from the sink towards the source. The~new heights are equal to the shortest distances towards the sink in the residual graph. The excesses of nodes, which are not availiable from the~sink in~backwards BFS tree, must be substracted from $ExcessTotal$, because it is a stored excess which will never reach the sink.
\\
\\
In 2011 Hong and He improved both heuristics to a new Asynchronous Global Relabeling (ARG) method ([7]). So far, the global and gap relabeling heuristics were run independent from the lock-free algorithm (in CUDA implementation, to run the heuristics the control was returned to the CPU). The main reason for this was that the push and the relabel operations are mutually exclusive with the global and gap relabeling heuristcs (a critical moment is when both the heuristics and the relabel operation want to~set a~new height for the~same vertex). However, this problem does not occur for~the~non-lock-free versions of the push-relabel algorithms (which is contrary to our expectations).
In the new approach, the ARG heuristic is executed by~a~distinguished thread which corresponds to any vertex and runs periodicaly, while the other threads asychronously run push or relabel operations. It significantly improves the execution time of the algorithm. The only problem is that ARG needs to maintain a queue of the unvisited vertices whose size is $O(V)$. In CUDA programming, a queue of that size can be maintained only in the global memory the access to which is very slow. Perhaps this is why the implementation presented in [7] uses~C and~the~\textit{pthread} library for~multi-threaded constructions.

	\subsection{Our Implementation}

In our implementation of the lock-free push-relabel algorithm we have tried to use the ARG heuristic but from the reasons mentioned above this algorithm turns out to be slower than the algorithm using a heuristic launched on CPU. Therefore we use the approach presented in Algorithm 4.8 with an~additional improvment. In the end of the global relabeling phase we add the~gap~relabeling heuristic which for each unvisited node in the BFS tree sets its height to $|V|$.

We implement the procedure of the lock-free push-relabel algorithm\\ as~a~CUDA~kernel. It is executed by $|V|+2$ threads. After \textit{CYCLE} iterations of the while loop in the kernel, the control is returned to the Host thread.
The constant \textit{CYCLE} is set to 7000 by a preprocessor macro (in our tests this value yielded best results). 
Next the host thread calls the C procedure running on CPU which performs a global relabeling.
The algorithm runs until all the flow of the value \textit{ExcessTotal} gets to the sink, then~the~excess of~the~sink is equal \textit{ExcessTotal}. Note that during the execution of~the~heuristic, the value of \textit{ExcessTotal} can be decreased. 

In our implementation a vertex is a structure \textit{node} holding four pointers: \textit{excess}, \textit{height}, \textit{toSourceOrSink} and \textit{firstOutgoing}. 
The pointer \textit{toSorceOrSink} points to the edge leading directly towards the source or the sink. This~makes the access to the source and the sink faster.
The pointer \textit{firstOutgoing} points to~the~first edge on~the~adjacency list of~the~vertex. The structure of an edge (named \textit{adj}) also holds four attributes. The pointer \textit{vertex} points to~the~neighbor to which the edge leads. The pointer \textit{flow} points to~the~location in the global memory where the residual capacity of the edge is stored. The attribute \textit{mate} is a pointer to the backward edge in the residual graph and \textit{next} is a pointer to the next edge on the adjacency list. 

During the push operation peforming on the edge $(x, y)$ both the residual capacities of edges $(x, y)$ and $(y, x)$ are changed. Then to improve the cache utilization, the residual capacities of edges $(x, y)$ and $(y, x)$ are stored one after another.

Throughout the algorithm only flows, excesses and heights need to be copied between the device and host. Therefore to minimize the data transfer between the host and the device we separate arrays containing these data from the structures.
Since global relabel heuristic gets all these arrays at~the~beginning than they are copied onto the host. On the other hand after performing the heuristic only the heights need to be copied back to~the~device. Therefore we store the excesses and the residual capacities in~a~single array and the hights are stored in a separate array. 

Before starting the algorithm the two arrays are allocated on the device containg nodes and edges respectively. The first keeps the structures of~type \textit{node} and the second contains the structures of type \textit{adj}. During the execution the algorithm arrays of nodes and edges are not copied to the host.

\section{Cost Scaling Algorithm}

It is known that the non-weighted matching problem can be easily reduced to the max flow problem (for more details see [1, paragraph "Maximum bipartite matching"]). In [9] Goldberg and Kennedy present a way how to efficiently solve the weighted matching problem with a cost scaling algorithm. In their work they reduce the assignment problem to \textit{the transporation problem} and present their implementation of the algorithm. 

We present an analogous reduction from the~assignment problem to~the~max flow min cost problem. In [9], for a given graph $G' = (V'=X'\cup Y', E')$, they additionaly define a \textit{supply} $d(x),\ x \in V$ that $\forall x \in X,\ d(x) = 1$ and~$\forall y \in Y,\ d(y) = -1$. It stimulates preflow pushed from the source to~every~node of $X$ and preflow pushed from every node of $Y$ to the sink. Hence the push relabel algorithm starts execution with $e(x) = d(x)$. In our work, instead of defining the supply and the transporation problem, we initialize the push-relabel algorithm with $e(x) = 1,\ x \in X$ and $e(x) = -1,\ y \in Y$.

Let $I$ be an instance of the assignment problem: $I = (G, w),\ G=(V=X\cup Y,\ E)$ where $|X|=|Y|$ and $w$ is a weight function for edges. We construct an instance $I' = (G', u, c)$ of the max flow min cost problem as~follows. For~each~edge $(x, y) \in E$ we add $(x, y)$ and $(y, x)$ to $E'$. For each $(x, y) \in X\times Y$ define capacities: $u(x, y) = 1$ and $u(y, x) = 0$, and costs: $c(x, y) = w(x, y)$ and $c(y, x) = -w(x, y)$. The graph $G'$ is still bipartite. 

\begin{figure}[h!b!]
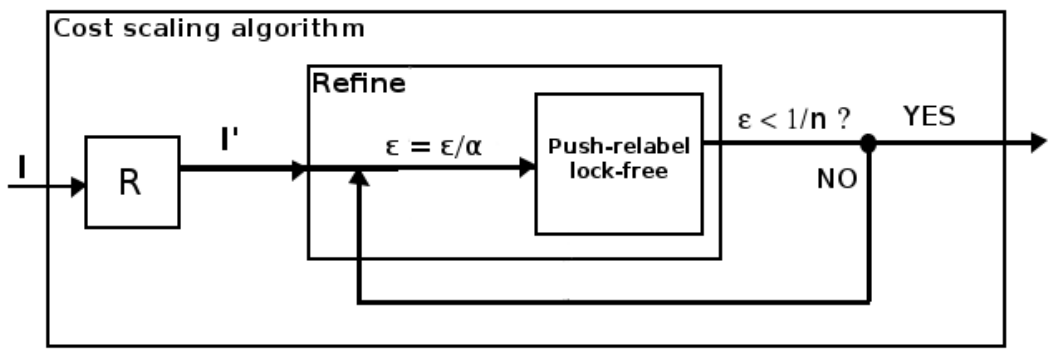

\begin{center}
\includegraphics[width=5.5in]e
\end{center}
\caption{Cost scaling algorithm.}
\end{figure}

		\subsection{Sequential Cost Scaling Algorithm}

The two cost scaling algorithms with efficient implementations which are described by Goldberg et al. in [8,9], are slightly different. The first of them~[8] is the generic cost scaling algorithm and was proposed by Goldberg in [10]. The second of them [9] was applied to the assignment problem.

Now we present the first version (Algorithm 5.0), next we point out differences from the second which requires some changes in the definitions of~the~$\epsilon$-optimal pseudoflow and the admissible edge. Our version of the algorithm uses the unmodified definitions and will be presented in subsection~5.3.

Let $C$ be the largest cost of an edge in G.
\\
\\
\begin{scriptsize}
\begin{tabular}{l}

\hline

\textbf{Algorithm 5.0. \textit{The cost scalling algorithm}} \\

\hline\\

1. \textit{Min-Cost():}\\[3pt]
2. $\epsilon \leftarrow C$\\
3. for each $x \in V$ do\\
4. \hspace{0.5cm}$p(x) \leftarrow 0$\\
4. while ($\epsilon \geq 1/n$) do\\
5. \hspace{0.5cm} $(\epsilon, f, p) \leftarrow$ Refine($\epsilon, p$)\\[10pt]

1. \textit{Refine($\epsilon, p$):}\\[3pt]
2. $\epsilon \leftarrow \epsilon / \alpha$\\
3. for each $(x, y) \in E:\ c_p(x, y) < 0$ do /*E contains edges of $G$, not $G_f$*/\\
4. \hspace{0.5cm}$f(x, y) \leftarrow 1$\\
5. while (f is not a flow) do\\
6. \hspace{0.5cm}apply a \textit{push} or a \textit{relabel} operation\\
7. return ($\epsilon$, $f$, $p$)\\[10pt]

1. \textit{push(x, y):}\\[3pt]
2. $\delta \leftarrow \min{\{e(x), u_f(x, y)\}}$\\
2. send $\delta$ units flow from $x$ to $y$\\[10pt]

1. \textit{relabel(x):}\\[3pt]
2. $p(x) \leftarrow max_{(x, z) \in E_f}\{p(z) - c(x, z) - \epsilon\}$\\

\\[3pt]\hline

\end{tabular}
\end{scriptsize}
\\
\\

The main procedure of the algorithm is a \textit{Min-Cost} method which maintains a variable $\epsilon$, a flow $f$, and a price function $p$. These variables are changed in the while loop by the procedure \textit{Refine}. On~the~beginning of~each step of the loop the flow $f$ is $\epsilon$-optimal with respect to $p$. The while loop stops when $\epsilon < 1/n$. It was proved that if \textit{Refine} reduces the parameter $\epsilon$ by~a~constant factor, the total number of iterations is $O(m \log n)$ (paper [10], Theorem 4.5).

The procedure \textit{Refine} starts with decreasing $\epsilon$ to $\epsilon / \alpha$ and saturating every admissible edge $(x, y)$, ie. $c_p(x, y) < 0$. This spoils the initial flow $f$ such that $f$ becomes an $\epsilon$-optimal pseudoflow, for $\epsilon = 0$ (because $\forall (x, y) \in E,\\ c_p(x, y) \ge 0$). This makes also some nodes active and some with a negative excess. Next the pseudoflow $f$ become an $\epsilon$-optimal flow by making a~series of the flow and the price update operations, each of which preserves $\epsilon$-optimality. There are two kind of the update operations: \textit{push} and \textit{relabel}.

The \textit{push(x, y)} operation is applied to an active node $x$ and a residual edge $(x, y)$ that is admissible: $c_p(x, y) < 0$. Then $\delta$ units of the flow is put to~the~node $y$: decreasing $e(x)$ and~$f(x, y)$ by $\delta$ and~increasing $e(y)$ and $f(y, x)$ by $\delta$.

The \textit{relabel(x)} operation is applied to an active node $x$ such that $(x, y)~\in~E_f$ and~$(x, y)$ does not satisfy the admissible constraints, i.e \mbox{$c_p(x, y) \geq 0$}. The~new value of~$p(x)$ is the smallest value allowed by the $\epsilon$-optimality constraints, ie. $max_{(x, z) \in E_f}{\{p(z) - c(x, z) - \epsilon\}}$.

Goldberg proved that the \textit{Refine} procedure is correct, that is, if it terminates, the pseudoflow $f$ is an $\epsilon$-optimal flow. Hence the min-cost algorithm is also correct (see [10]).

The procedure \textit{Refine} maintains a set $S$ containing all the active nodes.
The loop terminates when S becomes empty. The generic implementations of the procedure \textit{Refine} runs in $O(n^2 m)$ time, giving\\ an~$O(n^2 m \min\{\log(nC), m\log n\})$ time bound for computing a minimum-cost flow.

The differences between the algorithms mentioned above occur in a new definition of an admissible edge, in the initialization stage of the refine procedure, and in the relabel operation. 

Let $\epsilon > 0$. A residual edge $(x, y) \in E_f$ is \textbf{admissible}, if \mbox{$(x, y) \in {(X\times Y) \cap E_f}$} and $c_p(x, y) < \frac{1}{2} \epsilon$ or $(x, y) \in {(Y\times X) \cap E_f}$ and $c_p(x, y) < -\frac{1}{2} \epsilon$.
Then we have different conditions for the two types of edges.

The $\epsilon$-optimality notation is closely related to the admissible edge definition. Therefore after changing it we must later update the former, as~well. 
For a constant $\epsilon \geq 0$ a pseudoflow $f$ is \textbf{$\epsilon$-optimal} with respect to~a~price function $p$ if, for every residual edge $(x, y) \in {(X \times Y) \cap E_f}$, we have $c_p(x, y) \geq 0$, and for every residual edge $(y, x) \in {(Y \times X) \cap E_f}$, we have $c_p(y, x) \geq -\epsilon$. 
For a constant $\epsilon \geq 0$ a pseudoflow $f$ is $\epsilon$-optimal if it is $\epsilon$-optimal with respect to some price function $p$.
\\
\\
\begin{scriptsize}
\begin{tabular}{l}

\hline

\textbf{Algorithm 5.1. \textit{The cost scaling algorithm, version 2}} \\

\hline\\

1. \textit{Refine($\epsilon$, $p$):}\\[3pt]
2. $\epsilon \leftarrow \epsilon / \alpha$\\
3. for each $(x, y) \in E$\\  
4. \hspace{0.5cm}$f(x, y) \leftarrow 0$\ \ \ \ \ \ \ \ \ \ /*it can make some node active*/\\
5. for each $x \in X$\\
6. \hspace{0.5cm}$p(x) \leftarrow -min_{(x,y) \in E} c_p'(x, y)$\\
7. while (f is not a flow) do\\
8. \hspace{0.5cm}apply a \textit{push} or a \textit{relabel} operation\\
9. return ($\epsilon$, $f$, $p$)\\[10pt]

1. \textit{relabel(x);}\\[3pt]
2. if $x \in X$ then\\
3. \hspace{0.5cm}$p(x) \leftarrow max_{(x, y) \in E_f}{\{p(y) - c(x, y)\}}$\\
4. else if $x \in Y$ then\\
5. \hspace{0.5cm}$p(x) \leftarrow max_{(z, x) \in E_f}{\{p(z) + c(z, x) - \epsilon\}}$
\\[3pt]\hline

\end{tabular}
\end{scriptsize}
\\
\\
\\

On the start of the procedure refine, $\epsilon$ and an $\epsilon$-optimal flow $f$ are given. Similary as in the first algorithm, $\epsilon$ is decreased to $\epsilon / \alpha$. Next, for~each saturated edge $(x, y)$, the flow is removed from $(x, y)$ back to its node $x$. This makes some nodes active and some nodes obtain negative excess: $e(x) < 0$. Before the loop starts, for each node $x \in X$, $p(x)$ is set to~$-min_{(x,y) \in E} c_p'(x, y)$. Then, we get an $\epsilon$-optimal pseudoflow $f$, for $\epsilon = 0$ (because of $\forall{(x,y) \in E_f\cap (X\times Y)}\ \ c_p(x, y) \ge 0$ and $\forall{(y,x) \in E_f\cap (Y\times X)}\\ c_p(y, x)\ge -\epsilon$). Further, the push and the relabel operations are performed to make $f$ an $\epsilon$-optimal flow. 

It is easy to see that the differences between the algorithms do not result in any change in the returned output but they have impact on the efficiency. Therefore our idea is a combination these two codes (Algorithm 5.2.). We assume the first version of the definitions of admissible edges and $\epsilon$-optimal pseudoflows.
\\
\\
\begin{scriptsize}
\begin{tabular}{l}

\hline

\textbf{Algorithm 5.2. \textit{The cost scalling algorithm, our approach}} \\

\hline\\

1. \textit{Refine($\epsilon, p$):}\\[3pt]
2. $\epsilon \leftarrow \epsilon / \alpha$\\
3. for each $(x, y) \in E$\\
4. \hspace{0.5cm}$f(x, y) \leftarrow 0$\ \ \ \ \ \ \ \ \ \ /*it can make some node active*/\\
5. for each $x \in X$\\
6. \hspace{0.5cm}$p(x) \leftarrow -min_{(x,y) \in E}{\{c_p'(x, y) + \epsilon \}}$ /*it makes pseudoflow $f$ $\epsilon$-optimal, (also $0$-optimal)*/\\
7. while (f is not a flow) do\\
8. \hspace{0.5cm}apply a \textit{push} or a \textit{relabel} operation\\
9. return ($\epsilon$, $f$, $p$)\\[10pt]

1. \textit{relabel(x):}\\[3pt]
2. $p(x) \leftarrow -min_{(x, z) \in E_f}\{c_p'(x, z) + \epsilon\}$\\

\\[3pt]\hline

\end{tabular}
\end{scriptsize}
\\
\\

The initialization of the procedure refine (lines 2-6) is the~same as~in~Algorithm 5.1 except line 6, where we must adapt it to the appropriate as~definition the admissible edge. The relabel operation is the same as in Algorithm 5.0.

		\subsection{Heuristics: Price Updates and Arc Fixing}

The papers [8] and [9] provide a lot of inprovements to the cost scaling algorithm for its sequential version. We focus our attention on the price updates and arc fixing heuristics.
The idea of the price updates heuristic (introduced in [10] and described also in [15]) is similar to Dijkstra's shortest path algorithm, implemented using buckets as in Dial's implementation [16] of the shortest path algorithm.
\\
\\
\begin{scriptsize}
\begin{tabular}{l}

\hline

\textbf{Algorithm 5.3. Price updates heuristic} \\

\hline\\

\textit{price-updates-heuritic():}\\[3pt]
make empty sets $B[\cdot]$ /* buckets of unscanned nodes*/\\
make empty set $S$  /* set of active nodes */\\
make array $l$ /* constisted the labels of nodes*/\\
for each $x \in V$\\
\hspace{0.5cm}$x$.scanned $\leftarrow false$\\
\hspace{0.5cm}$l(x) \leftarrow \infty$\\
\hspace{0.5cm}if $e(x) < 0$ then\\
	\hspace{1.2cm}$B[0]$.push($x$) /* bucket $0$ is consisted of the $x$ which excess is negative*/\\
	\hspace{1.2cm}$x$.bucket $\leftarrow 0$\\
\hspace{0.5cm}else\\
	\hspace{1.2cm}$x$.bucket $\leftarrow \infty$\\
	\hspace{1.2cm}if ($e(x) > 0$) then\\
		\hspace{2.0cm}$S$.push($x$) /* active node $x$ is pushed to $S$ */\\
$i \leftarrow 0$\\
while ($S$ is not empty) do\\
\hspace{0.5cm}while ($B[i]$ is not empty) do\\
	\hspace{1.2cm}$x \leftarrow B[i]$.pop()\\
	\hspace{1.2cm}for each $(y, x) \in E_f$ do\\
		\hspace{2.0cm}if ($y$.scanned $= false$ and $\lfloor c_p(y, x) / \epsilon \rfloor + 1 < y$.bucket) do \\
			\hspace{2.5cm}$old \leftarrow y$.bucket\\
			\hspace{2.5cm}$new \leftarrow \lfloor c_p(y, x) / \epsilon \rfloor + 1$\\
			\hspace{2.5cm}$y$.bucket $\leftarrow new$\\
			\hspace{2.5cm}$B[old]$.pop($y$)\\
			\hspace{2.5cm}$B[new]$.push($y$)\\
	\hspace{1.2cm}$x$.scanned $\leftarrow true$\\
	\hspace{1.2cm}$l(x) \leftarrow i$\\
	\hspace{1.2cm}if ($e(x) > 0$) then\\
		\hspace{2.0cm}$S$.pop($x$)\\
	\hspace{1.2cm}$last \leftarrow i$\\
\hspace{0.5cm}$i \leftarrow i+1$\\
for each $x \in V$\\
\hspace{0.5cm}if $l(x) < \infty$\\
	\hspace{1.2cm}$p(x) \leftarrow p(x) - \epsilon\ l(x)$\\
\hspace{0.5cm}else\\
	\hspace{1.2cm}$p(x) \leftarrow p(x) - \epsilon\ (last + 1)$\\
\\[3pt]

\hline

\end{tabular}
\end{scriptsize}
\\
\\

In our implementation, \textit{the price update heuristic} maintains the~set of~buckets $B[\cdot]$. Each node belongs to at most one bucket. Any~node with~a~negative excess is in $B[0]$ and initially buckets with the number larger than~$0$ are~empty. Therefore, each node with a non-negative excess belongs to~number~$\infty$. In general, for any node the number of its bucket~equals its~distance from any node residing in the 0-bucket. Additionaly, the heuristic mainatains an array $l$ containing of the labels, which are distances to~some~node with~negative excess i.e. a node residing in the 0-bucket. The distances are multiplies of $\epsilon$.

In each iteration of algorithm, a node with the lowest distance (ie.~a~node residing in a nonempty bucket with the lowest number) is scanned. After~the~scanning, the distance of that node is set to the number of its bucket and the node is marked as scanned. 

During the scanning, residual edges entering the scanned node are tested. If~a~neighbor was not scanned yet and its new calculated distance is smaller than the present, then its distance is updated to a new and this~neighbor is~removed from the old bucket and put into to a new bucket.
\\
\\

In our implementation we also use the \textit{arc heuristic} described in [8]. It~deletes some edges from a residual graph thus decreasing the running time.
For the $\epsilon$-optimal flow $f$ and edge $e$, if $c_p(e) > 2n\epsilon$ then~the~flow of~$e$~will never be changed. Therefore this edge can be permanently omitted.

	\subsection{Applying Lock-Free Push-Relabel Method to Cost Scaling Algorithm}

To get better running time of the \textit{Refine} procedure than in a sequential algorithm, we improve it by applying the lock-free push-relabel algorithm (Algorithm 5.4). 
\\
\\
\begin{scriptsize}
\begin{tabular}{l}
\hline

\textbf{Algorithm 5.4. \textit{Push-relabel, our approach}} \\

\hline\\

\textit{refine($G, \epsilon, p$):}\\[3pt]
1. /* x - node operated by the x thread */\\
2. while ($e(x) > 0$) do\\
3. \hspace{0.5cm}$e' \leftarrow e(x)$\\
4. \hspace{0.5cm}$\tilde{y} \leftarrow NULL$\\
5. \hspace{0.5cm}$min\_c_p' \leftarrow \infty$\\
6. \hspace{0.5cm}for each $(x, y) \in E_f$ do\\
7.        \hspace{1.2cm}$tmp\_c_p' \leftarrow c_p'(x, y)$\\
8.        \hspace{1.2cm}if $min\_c_p' > tmp\_c_p'$ do\\
9.        	\hspace{2.0cm}$min\_c_p' \leftarrow tmp\_c_p'$\\
10.		\hspace{2.0cm}$\tilde{y} \leftarrow y$\\
11.\hspace{0.5cm}if ($min\_c_p' < -p(x)$) do /* that is: $c_p(x, \tilde{y}) < 0 $ - edge is admissible */ \\
12.	\hspace{1.2cm}/* then the x thread performs PUSH towards $\tilde{y}$ */\\
13.	\hspace{1.2cm}$u_f(x, \tilde{y})\ \leftarrow u_f(x, \tilde{y}) - 1$\\
14.	\hspace{1.2cm}$u_f(\tilde{y}, x)\ \leftarrow u_f(\tilde{y}, x) + 1$\\
15.	\hspace{1.2cm}$e(x) \leftarrow e(x) - 1$\\
16.	\hspace{1.2cm}$e(\tilde{y}) \leftarrow e(\tilde{y}) + 1$\\
17.\hspace{0.5cm}else do / *then the x thread performs RELABEL */\\
18.	\hspace{1.2cm}$p(x) \leftarrow -(min\_c_p' + \epsilon)$\\
\\[3pt]\hline

\end{tabular}
\end{scriptsize}
\\
\\

As in Hong's algorithm, we assume that one node can be operated by~at~most one thread. Each thread has the following private atributes. The~variable $e'$ stores the excess of the node $x$. The variable $\tilde{y}$ stores the~node with~the~lowest reduced cost. The variable $min\_c_p'$ stores the lowest partially reduced cost of the edge $(x, \tilde{y})$. The variable $tmp\_c_p'$ stores the temporarily reduced cost of the edge $(x, y)$. The arrays with excesses $e$ and prices $p$ of~nodes, and costs $c$ and residual capacity $u_f$ of edges are shared beetwen all the~running threads.

Let $x$ be an active node operated by thread $x$. Before performing \textit{push} or \textit{relabel} operation, $x$ select the residual edge $(x, y) \in E_f$, whose reduced cost is the lowest among all residual edges outgoing from $x$ (lines 6-10). Let $(x, \tilde{y})$ be the edge with the reduced cost $min\_c_p'$. The thread $x$ verifies whether $(x, \tilde{y})$ is admissible (line 11). If so, $x$ puts the unit of flow towards $\tilde{y}$: decreasing excess of $x$ and residual capacity of $(x, \tilde{y})$, and increasing excess of $\tilde{y}$ and residual capacity of $(\tilde{y}, x)$ (lines 13-16). Otherwise, if $(x, \tilde{y})$ is not admissible, the relabel operation is performed on $x$ and new price of $x$ is $p(x) \leftarrow -(min\_c_p' + \epsilon)$ (line 18).

Obviosly, the decreasing and increasing variables in the push operation must be atomic because of the write conficts, which may occur when two threads will want to change excess of the same node.

		\subsection {Correctness of Algorithm}

The correctness of the algorithm follows from the correctness of the \textit{Refine} procedure.

The main difference between our algorithm and the push-relabel algorithm for the max flow problem is that in the second method for each node, each of its neighbors sees the same height of that node. In our algorithm for~each~node, each of its neighbors sees the same price of that node but~to~determine the~$\epsilon$-optimality edge, the reduced cost of an edge (not~the~price of the node) must be considered, which can be different for each neighbors.

Our proof is based on Hong's proof of the correctness of the lock-free push-relabel algorithm [5] and Goldberg and Tarjan's proof of the correctness of the generic refine subroutine [2].

\begin{lm}
During the execution of the refine operation, for any active node one of two operations: $push$ or $relabel$ can be applied.
\end{lm}

\begin{proof}
Let $x$ be an active node. Then there must be an edge $(x, y) \in E_f$ because a push operation has occured which increased the value of~$e(x)$ and~$u_f(x, y)$. Let $c_p(x, y) >= 0$. Hence $(x, y)$ is not admissible and a push operation is impossible. Therefore a relabel operation can be applied.
\end{proof}

\begin{lm}
During the execution of the refine operation the price of a node never increases.
\end{lm}

\begin{proof}
Let $x$ be an active node. If there is a residual edge $(x, y)$\\ such~that~$c_p(x, y)~<~0$ then the push operation can be applied to it and~the~price of~$x$~does not increase. 

Otherwise, if for each residual edge $(x, y)$, $c_p(x, y) >= 0$ then the relabel operation can be applied to $x$. Let $(x, \tilde{y}) \in E_f$ have the smallest reduced cost among all the residual edges outgoing from $x$. Since $c_p(x, \tilde{y}) >= 0$ then $c(x, \tilde{y}) + p(x) - p(\tilde{y}) >= 0$ and $p(x) >= p(\tilde{y}) - c(x, \tilde{y}) = - c_p'(x, \tilde{y})$. Additionaly, after the relabel operation we have $p(x) = -(c_p'(x, \tilde{y}) + \epsilon)$. Therefore the price of $x$ does not increase.
\end{proof}

Now, following Hong's proof, we define a trace, a preparation stage and~a~fulfillment stage and two basic types of computation (history) traces consisting of push and/or relabel operations.

The \textit{trace} of the interleaved execution of multiple threads is the~order in~which instructions from the threads are executed in real time.

Each operation (both push or relabel), can be split into 2~stages: the~preparation and the fulfillment. 

For the push operation the preparation stage is performed in lines 6-11 of~Algorithm 5.4 and the fulfillment stage is lines 12-16.
For the relabel operation the preparation stage is in lines 6-10 of Algorithm 5.4 and~the~fulfillment stage is in line 18.

The preparation stage tests if the operation is aplicable. The fulfillment stage finishes the operation. The P and F notations denote the preparation and the fulfillment stage respectively.

A \textit{stage-clean} trace consists of non overlapping operations, for example: (P(push1), F(push1), P(relabel1), F(relabel1), P(push2), F(push2)).\\
In \textit{stage-stepping} trace all the preparation stages are performed before any~fulfillment stage, for example: (P(push1), P(relabel2), P(relabel1), F(relabel2), F(push1), F(relabel1)).

\begin{lm}
Each trace consisting of push and/or relabel operations is semantically equivalenty to one of the two basic traces: stage-clean or stage-stepping.
\end{lm}

\begin{proof}
There are 5 nontrivial cases in which the stages of push and relabel operations on common data can interleave:

\begin{enumerate}
\item \textit{push(x, y)} and \textit{push(y, z)}

The fulfillment stage of $push(x, y)$ increases $e(y)$. However the preparation stage of $push(y,z)$ reads $e(y)$. The reading and the writing operations can occur in the following three scenarios.

\begin{enumerate}
\item $P(push(x, y))$ - $F(push(x, y))$ - $P(push(y, z))$.

This is a stage-clean trace: push(x, y) -> push(y, z). In~the~fulfillment of $push(x, y)$ the thread $x$ increases $e(y)$ before~the~thread $y$ reads $e(y)$ in the preparation stage.

\item $P(push(y, z))$ - $F(push(y, z))$ - $F(push(x, y))$.

This is a stage-clean trace: push(y, z) -> push(x, y). In~the~fulfillment of $push(y, z)$ thread $y$ decreases $e(y)$ before~the~thread $x$ writes to the $e(y)$ in fulfillment stage.

\item $P(push(y, z))$ - $F(push(x, y))$ - $F(push(y, z))$.

This is also a stage-clean trace: push(y, z) -> push(x, y). In~the~fulfillment of $push(x, y)$, the thread $x$ increases $e(y)$, which was earlier remembered by thread $y$. Next, thread $y$ pushes the flow, whose value is depended from the old value of $e(y)$, to the node $z$. So it is equvalent that the thread $y$ pushes the flow to $z$ and~then~the~push operation from $x$ to $y$ is performed.

\end{enumerate}

\item \textit{push(x, y)} and \textit{push(z, y)}

Both $push(x,y)$ and $push(z,y)$ increase the value of $e(y)$ without reading and storing this value before. It may occur in two different scenarios, each of them is equivalent to a stage-clean trace.

\item \textit{push(x, y)} and \textit{relabel(y)}

In the preparation of $push(x,y)$ the thread $x$ chooses the edge $(x, y)$ because it has the lowest partially reduced cost $c_p'(x, y)$. However, in~the~fulfillment stage the thread $y$ updates the price $p(y)$. We have six different scenarios interleaving the stages of $push(x, y)$ and $relabel(y)$.

\begin{enumerate}

	\item $P(push(x, y))$ - $F(push(x, y))$ - $P(relabel(y))$ - $F(relabel(y))$

Obvious. It is equivalent to the stage-clean trace: $push(x, y)$ -> $relabel(y)$.

	\item $P(relabel(y))$ - $F(relabel(y))$ - $P(push(x, y))$ - $F(push(x, y))$

Obvious. It is equivalent to the stage-clean trace: $relabel(y)$ -> $push(x, y)$.

	\item 	$P(push(x, y))$ - $P(relabel(y))$ - $F(push(x, y))$ - $F(relabel(y))$\\
		$P(push(x, y))$ - $P(relabel(y))$ - $F(relabel(y))$ - $F(push(x, y))$\\
		$P(relabel(y))$ - $P(push(x, y))$ - $F(relabel(y))$ - $F(push(x, y))$\\
		$P(relabel(y))$ - $P(push(x, y))$ - $F(push(x, y))$ - $F(relabel(y))$

In this cases two preparation stages are performed before any fulfillment stage one. Therefore they are equivalent to the stage-stepping~trace: $P(push(x, y)$ -> $P(relabel(y))$ -> $F(push(x, y))$ -> $F(relabel(y))$.

\end{enumerate}

\item \textit{push(x, y)} and \textit{relabel(z)}

Let $\tilde{p}(z)$ be the new price of $z$ after the fulfillment stage of $relabel(z)$.
According to lemma 4.2,: $\tilde{p}(z) \le p(z)$ must hold.\\ 
We have $c_p(x, y) = \min_{(x, w) \in E_f}c_p(x, w) \le c_p(x, z)$.\\ 
Then $c_p(x, z) = c(x, z) + p(x) - p(z) \le c(x, z) + p(x) - \tilde{p}(z) = \widetilde{c_p}(x, z)$.
Since $\widetilde{c_p}(x, z) \ge c_p(x, z) \ge c_p(x, y) \ge 0$, then the operation $relabel(z)$ does not impact the operation $push(x, y)$ and this scenario is equivalent to a stage-clean trace.

\item \textit{relabel(x)} and \textit{relabel(y)}

In the fulfillment stage of $relabel(x)$ the thread $x$ updates $p(x)$ and~$(y, x)$ may be the residual edge read by the thread $y$ before or after the~fulfilment stage of $relabel(y)$. Then there are six scenarios interleaving the stages of $relabel(x)$ and $relabel(y)$. The proof is analogous to~the~case~3.

\end{enumerate}
\end{proof}

\begin{lm}
For any trace consisting of three or more push and/or relabel operations there exists an equivalent sequence consisting only of stage-clean or stage-stepping traces.
\end{lm}

\begin{proof}
Similar as above.  
\end{proof}

\begin{lm}
When the algorithm terminates $f$ is an $\epsilon$-optimal flow.
\end{lm}

\begin{proof}
We show that during the execution of the algorithm any residual edge $(x, y)$ satisfies the constraint $c_p(x, y) \ge -\epsilon$ with one exception which is transient.
Let $f$ be an $\epsilon$-optimal pseudoflow. There may occur the following situations:

\begin{enumerate}

\item applying the push operation

Obviously.

\item applying the relabel operation

Let $x$ be the relabeled node and $\tilde{p}(x)$ be the new price of $x$. According to Lemma 5.2, we have $\tilde{p}(x) < p(x)$.
Then any edge of the form: $(x, y)$ or $(y, x)$ can spoil the $\epsilon$-optimality of $f$. For any residual edge $(x, y)$ outgoing from $x$, we have $c_p(x, y) = c(x, y) + p(x) - p(y) < c(x, y) + \tilde{p}(x) - p(y) = \widetilde{c_p}(x, y)$ Therefore, if $\widetilde{c_p}(x, y) \ge c_p(x, y) \ge -\epsilon$ then the relabel operation preserves the $\epsilon$-optymality.

\item applying $relabel(x)$ and $relabel(y)$ operations

According to Lemma 5.3 this scenario can be reduced to a stage-clean trace or stage-stepping trace. A clean-stage trace can be reduced to~cases: 1. and 2. above. A stage-stepping trace has the following cases.

\begin{enumerate} 

\item $(x, y) \in E_f \text{ and } (y, x) \in E_f$
	
In this case, $c_p(x, z) >= 0$, for all $(x, z) \in E_f$, and $c_p(y, w) >= 0$, for all $(y, w) \in E_f$. 
Therefore, since\\\mbox{$c_p(x, y) = c(x, y) + p(x) - p(y) = -(c(y, x) + p(y) - p(x)) = -c_p(y, x)$} then $c_p(x, y) = -c_p(y, x) = 0$.

Furthermore, $c_p(x, y) = \min{\{c_p(x, z),\ (x, z) \in E_f\}}$ and\\ $c_p(y, x) = \min{\{c_p(y, w),\ (y, w) \in E_f\}}$.

Therefore, the fulfillment stages of $relabel(x)$ and $relabel(y)$ update the prices of $x$ and $y$ respectively:\\ $\tilde{p}(x) \leftarrow - (c_p'(x, y) + \epsilon)$ and $\tilde{p}(y) \leftarrow - (c_p'(y, x) + \epsilon)$.\\

Let us check whether the new prices preserve $\epsilon$-optimality of~the~residual edges $(x, y)$ and $(y, x)$. 
Since\\ $\widetilde{c_p}(x, y) = c(x, y) + \tilde{p}(x) - \tilde{p}(y) = c(x, y) - c_p'(x, y) - \epsilon + c_p'(y, x) + \epsilon = 
c(x, y) - c(x, y) + p(y) + c(y, x) - p(x) = c(y, x) + p(y) - p(x) = c_p(y, x) = 0 >= -\epsilon$, then $(x, y)$ preserves $\epsilon$ - optimality of~the~pseudoflow $f$. Because of $\widetilde{c_p}(x, y) = -\widetilde{c_p}(y, x)$ then also $c_p(y, x) = 0$ and $(y, x)$ preserves $\epsilon$-optimality of $f$. 

\item $(x, y) \in E_f \text{ and } (y, x) \notin E_f$

We want to see if the edge $(x, y)$ preserves $\epsilon$-optimality of~the~pseudoflow $f$ after the fulfillment stages of $relabel(x)$ and $relabel(y)$. 
Let $(x, z)$ be the edge that affects on the new price $\tilde{p}(x)$ of $x$. 

Since $c_p(x, z) \le c_p(x, y)$ then $c_p'(x, z) \le c_p'(x, y)$ and\\
$\tilde{p}(x) = -(c_p'(x, z) + \epsilon) \ge -(c_p'(x, y) + \epsilon) = -c(x, y) + p(y) - \epsilon \ge -c(x,y) + \tilde{p}(y) - \epsilon$.\\
Therefore $\widetilde{c_p}(x, y) = c(x, y) + \tilde{p}(x) - \tilde{p}(y) \ge - \epsilon$. Q.E.D

\item $(x, y) \notin E_f \text{ and } (y, x) \in E_f$
	
	The proof is the same as above.

\item $(x, y) \notin E_f \text{ and } (y, x) \notin E_f$
	
	If $(x, y)$ and $(y, x)$ are not residual then after the relabel operations they will not be either. Therefore, this case is trivial.
 
\end{enumerate}

\item applying $push(x, y)$ and $push(z, y)$ operations

Following Lemma 5.3, this case is equivalent a stage-clean trace hence it reduces to the cases: 1 and 2.

\item applying $push(x, y)$ and $relabel(y)$ operations

This case is equivalent to either a stage-clean or a stage-stepping trace. The first scenario can be reduced to the cases 1 and 2, while the second has the following cases. Note $(x, y) \in E_f$ because the push operation is applicable.

\begin{enumerate}

\item $(y, x) \in E_f$

Since the push operation preserves $\epsilon$-optimality of the pseudoflow $f$, then from the inductive assumption, the pseudoflow $f$ is \mbox{$\epsilon$-optimal} after the push operation.

The fulfillment stage of $push(x, y)$ does not affect the~value~of~$c_p(y, x)$.~If $0 \le c_p(y, x) \le \min{\{c_p(y, z),\ (y, z) \in E_f\}}$ before the~fulfillment~stage of $relabel(y)$ then the new price of the node $y$ satisfies $\tilde{p}(y) = -(\min_{(y,z)\in E_f}{c_p'(y,z)} + \epsilon) \ge - (c_p'(y, x) + \epsilon)$.
Therefore, after the fulfillment stage of $relabel(y)$ we obtain: $\widetilde{c_p}(y, x) = c(y, x) + \tilde{p}(y) - p(x) \ge c(y, x) - c_p'(y, x) - \epsilon - p(x) = c(y,x) - c(y,x) + p(x) - \epsilon -p(x) = -\epsilon$. Then the residual edge $(y, x)$ is $\epsilon$-optimal with~respect to the new price of the node $y$.

\item $(y, x) \notin E_f$

There are two subscenarios.

\begin{enumerate}

\item $(y, x) \in E$ 

In the situation that $(y, x) \in E$ and $(y, x) \notin E_f$ we have $f(y,x) = u_f(y,x)$. However, the fulfillment stage of $push(x, y)$ pushes the flow through the edge $(x, y)$ and hence adds the~reverse~edge $(y, x)$ to $E_f$. Otherwise, the preparation stage of~$relabel(y)$, which is performed before the fulfillment stage of $push(x, y)$, would not consider the edge $(y, x)$ (because it would be added to $E_f$ after that stage).
Let $(y, z_0)$ be the~edge that affects the new price of $y$ in the preparation stage of~$relabel(y)$. After the trace, it may occur that the~new residual edge $(y, x)$ has a lower reduced cost than $(y, z_0)$, that is $c_p(y, x) < c_p(y, z_0)$ and hence $c_p'(y, x) < c_p'(y, z_0)$. Therefore $-(c_p'(y, x) + \epsilon) > -(c_p'(y, z_0) + \epsilon)$ what implies that $-(c(y, x) - p(x) + \epsilon) > \tilde{p}(y)$ and hence
\mbox{$c(y, x) + \tilde{p}(y) -p(x) < -\epsilon$} and $c_p(y, x) < -\epsilon$. Then the residulal edge $(y, x)$ spoils the~\mbox{$\epsilon$-optimality} of~the~pseudoflow~$f$. Further we show that this situation is temporary and in the next step for the node $y$, $f$ becomes an $\epsilon$-optimal pseudoflow.
For the node $y$, $e(y) > 0$. Moreover, the residual edge $(y, x)$ has the smallest reduced cost among all the residual edges outgoing from $y$. Then the push operation towards $x$ can be performed on $y$ and $u_f(y, x) = \min{\{u_f(y, x), e(y)\}}$. Therefore, in the next step for the node $y$, the push operation removes the residual edge $(y, x)$ from the residual graph and hence remove the requirement $c_p(y, x) \ge -\epsilon$.

\item $(y, x) \notin E$

In the situation that $(y, x) \notin E$ and $(y, x) \notin E_f$ there must be $f(x, y) = 0$ (if $f(x, y) > 0$ then $u_f(y, x) = u(y, x) - f(y, x) = u(y, x) + f(x, y) > 0$ and then it would be $(y, x)~\in~E_f$). Similarly to the previous case, the residual edge $(y, x)$ will be added to $E_f$ in the fulfillment stage of $push(x, y)$ and the~\mbox{$\epsilon$-optimality} of the flow might be violated because of~$(y, x)$. However, like in the previous case, the residual edge $(y, x)$ will be removed from~$E_f$ because of the push operation from $y$ towards $x$ which will push the flow of value $f(y, x)$. 

\end{enumerate}

\end{enumerate}

\item applying $push(x, y)$ and $push(y, z)$ operations.

According to Lemma 5.3 this trace is equivalent to a stage-clean trace, where $push(x, y)$ and $push(y, z)$ are performed one by one. Then this case can be reduced to the case 2 above. The $\epsilon$-optimality is preserved by this trace.

\item applying $push(x, y)$ and $relabel(x)$ operations.

According to Lemma 5.2 this two operations cannot be intereleaved because the active node can be dealt with only by one operation per step of the algorithm.

\item applying $push(x, y)$ and $push(y, x)$ operations.

These two operations cannot be intereleaved because the two conditions $c_p(x, y) < 0$ and $c_p(y, x) < 0$ are mutually exclusive.

\item applying more operations than two.

Any trace constisting of more than two operations is interleaved.\\ The~proof is similar to the above and is omitted.

\end{enumerate}
\end{proof}

\begin{lm}
When the algorithm terminates the pseudoflow $f$ is $\epsilon$-optimal.
\end{lm}

\begin{proof}
The initial pseudoflow $f$ is $\epsilon$-optimal. According to lemma 5.5, any~trace consisting of the flow and/or the price update operations preserve \mbox{$\epsilon$-optimality} of $f$. According to lemma 5.1, the refine procedure can terminate only if~there is no active node, which implies that the pseudoflow $f$ is an $\epsilon$-optimal flow.
\end{proof}

	\subsection{Our Implementation}

Our implementation of the cost scaling algorithm uses the method presented in Algorithm 5.4, the global price update heuristic in Algorithm 5.3 and the arc-fixing heuristic described above.
We also use our implementation of the lock-free push-relabel algorithm changed appropriately for the~cost-scaling method. 

%In our implementation a vertex is the structure \textit{node}. \textit{node} holding the following elements. \textit{excess} points to the location in array containing the excesses of nodes. The pointer \textit{price} points to the location in array containning the prices of nodes. The \textit{set} stores the information to which of the two sets: $X$ or $Y$ (where $G=(V=X\cup Y, E)$) the node belongs. The \textit{toSorceOrSink} points to the edge towards the source or the sink (this makes the access to the source and the sink faster). The \textit{firstOutgoing} is a pointer to the first edge on the adjacency list of neighbors. 

%An edge is the structure \textit{adj}, which holds the following elements. The pointer \textit{vertex} points to a neighbor to which an edge leads. The \textit{flow} points to the residual capacity located in the global memory (maintained in exchange for maintaining the values of a flow and a capacity). The \textit{cost} stores the value of a cost. The \textit{mate} is a pointer to the a backward edge in the residual graph. The \textit{next} is a pointer to a next edge on the adjacency list. The \textit{prev} is a pointer to a previous edge on the adjacency list.

After loading the input, the graph is constructed in the host global memory. Next, all arrays are copied to the global memory on device. Since the~copied addresses point to locations in the host memory and the graph can be arbitrary we use an additional array to store information about adjacency lists of vertices and set valid addresses in the device memory. 
Futher, while running the algorithm only three arrays are copied between host and~device: the prices, the excesses and the flows.

In each iteration of the refine loop the host thread calls the push-relabel kernel until the pseudoflow $f$ becomes a flow, that is,\\ until $EXCESS(source)~=~0$ and $EXCESS(sink) = Total$,\\ where $Total~=~|X|~=~|Y|$. The kernels are launched by $|V| + 2$ threads. The control is returned back to the host in two cases. First, after each \textit{CYCLE} iterations, (\textit{CYCLE} has been preset to 500000).
In the second case, if~$EXCESS(source) = 0$ and $EXCESS(sink) = Total$. The first case means that the calculation is not finished, but the control is returned to~the~host. The second means that the pseudoflow $f$ becomes a flow. In~the~first case, it is common that the running kernel is interrupted and restored without changes in the graph. 
The purpose of this is to avoid terminating the execution of the kernel by GPU (\textit{launch timed out}). This often happens when the graph has many edges.

In our implementation only after the first running of the push-relabel kernel the heuristics are performed. They can be executed many times but the way we did it yields the best results. Similarly as in [15] we chose the~value~of~$ALPHA$ constant equal $10$ because in our tests other values much extended the running time of the program. The best results of the~algorithm are achived for the thread block of size $32 \times 16$.

The first heuristic executed is $arc\_fixing$ kernel then $globalPriceUpdate$ C procedure. The $arc\_fixing$ kernel deletes edges from the graph residing in the device memory by removing them from the adjacency lists and sets the~flows on them to $-10$ (any negative constant would do) in order not to free the memory. Then the arrays of~flows and prices are copied to the host memory. To make the graphs on~the~device and on the host equal, the procedure of the global relabeling deletes from the host graph the edges which were previously deleted from~the~graph on the device memory. They are recognized by the value of~a~flow equal $-10$. Next the global relabeling is performed and further the~new prices are copied back to the device memory.

	\section{Discussion and Future Work}

In this paper we have presented an efficient implementation of Hong's parallel non-blocking push-relabel algorithm and an~implementation of~the~cost-scaling algorithm for~the~assignment problem. The \textit{Refine} procedure of~the~cost-scaling algorithm uses the parallel push-relabel lock-free algorithm and runs in $O(n^2 m)$ time, where the complexity is analyzed in respect to the number of push and relabel operations. The amount of work performed by the program (the total number of instructions executed) is the same as in the sequential Algorithm 5.2, i.e. $O(n^2m \min{\{\log(nC), m\log n\}})$. Both implementations run on GPU using CUDA architecutre. 

For our purpose, the cost scaling algorithm for the assignment problem is used for the complete bipartite graphs of size $|X|=|Y|<=30$. 
The~execution time of our implementation for graphs of this size and costs of~edges at~most~$100$, is about $1/20$s which allows for real-time applications. 
%The~execution time increases to $0.185s$ when the costs increase to $100000$.

The maximum efficiency of an algorithm implemented in CUDA architecture is achieved when the number of threads is large and all threads execute the same task. Otherwise the running time can be equal the running time of a sequential algorithm. In the case of the large complete bipartite graphs the presented algorithm is not efficient.

Further reasearch could be made towards the use of the CUDA architecutre to finding a maximum cardinality matching of maximum weight for arbitrary (nonbipartite) graphs.

\end{document}